\documentclass[11pt]{article}
\usepackage{rheaj}
\usepackage{setspace}

\renewcommand{\mypara}[1]{\medskip \noindent {\bf #1}}

\newcommand{\IS}{\textnormal{IS}}
\newcommand{\polyloglog}{\textnormal{polyloglog}}
\title{Exponential Time Approximations for \\ Coloring 3-Colorable Graphs}

\author{
    Venkatesan Guruswami\thanks{Simons Institute for the Theory of Computing and the Dept.\ of EECS, Univ.\ of California, Berkeley, Berkeley, CA 94720. \texttt{venkatg@berkeley.edu}. Research supported in part by a Simons Investigator award and NSF grant CCF-2211972.}
    \and
    Rhea Jain\thanks{Dept.\ of Computer Science, Univ.\ of Illinois,
    Urbana-Champaign, Urbana, IL 61801. \texttt{rheaj3@illinois.edu}.}}

\date{}

\begin{document}

\maketitle

\begin{abstract}
  The problem of efficiently coloring 
  $3$-colorable graphs with few colors has received much 
  attention on both the algorithmic and inapproximability fronts. 
  We consider exponential time approximations, in which given a parameter 
  $r$, we aim to develop an $r$-approximation algorithm with the best 
  possible runtime, providing a tradeoff between runtime and approximation
  ratio. In this vein, an algorithm to $O(n^\eps)$-color a 3-colorable
  graphs in time $2^{\Theta(n^{1-2\eps}\log(n))}$ is given in (Atserias and Dalmau, SODA 2022.)

  \smallskip
  We build on tools developed in (Bansal et al., Algorithmic, 2019) 
  to obtain an algorithm to color $3$-colorable graphs with $O(r)$ colors in 
  $\exp\left(\tilde{O}\left(\frac {n\log^{11/2}r} {r^3}\right)\right)$ time, 
  asymptotically improving upon the bound given by Atserias and Dalmau.
\end{abstract}

\section{Introduction}
\label{sec:intro}

Graph Coloring is a classic and fundamental problem in theoretical computer 
science with applications in scheduling, register allocation, and more. 
A graph is said to be $r$-colorable if there is an assignment of $r$ colors to 
vertices such that no two adjacent vertices share the same color. 
The question of whether a given graph can be colored using $r$ colors for 
$r \geq 3$ was one of Karp's NP-complete problems \cite{Karp1972}. 
For any graph $G$, the minimum $r$ such that $G$ is $r$-colorable is called 
the \emph{chromatic number} of $G$; it is known that the chromatic number is 
hard to approximate to a factor $n^{1-o(1)}$ \cite{FeigeK98,Khot01}.

It is natural to ask if one can obtain a valid coloring of a graph with few colors 
if its chromatic number is known. 
The first polynomial time algorithm for this problem was given by 
Wigderson \cite{wigderson1983}, who obtained an $O(n^{1 - \frac 1 {r-1}})$-coloring 
for $r$-colorable graphs. In this paper we focus on the special case where $G$
is $3$-colorable. There are several known polynomial time algorithms for this setting.
Blum obtained a polynomial improvement over \cite{wigderson1983} to 
$\tilde O(n^{3/8})$ colors \cite{blum1994} when the input graph is $3$-colorable.
This was improved to $O(n^{1/4})$ using semi-definite programming methods 
by Karger, Motwani, and Sudan \cite{karger1994approximate}.
A long line of subsequent work culminated in the the development of 
an algorithm using
$O(n^{0.19996})$ colors \cite{KawarabayashiT17}, very recently improved 
to $O(n^{0.19747})$ \cite{KTY24}. 
On the hardness of approximation side, it is known to be NP-hard to color 
a 3-colorable graph using 4 colors 
\cite{KhannaLS00,venkat00hardness,brakensiek_venkat16}, 
and until recently, no better lower bound was known. 
Barto et al. recently showed that it is NP-hard to color an 
$r$-colorable graph using $2r-1$ colors; in particular, it is NP-hard to 
5-color a 3-colorable graph \cite{blko21}. 
Furthermore, there is some evidence based on 
variants of the Unique Games Conjecture that coloring an $r$-colorable graph 
with $O(1) \cdot r$ colors is hard for any $r \geq 3$ \cite{DinurMR09,GSandeep00}.

The hardness of approximating graph coloring in polynomial time has motivated the 
design of algorithms that provide a more fine-grained trade-off between running 
time and approximation factors; see 
\cite{bonnet2018sparsification,bourgeois2011approximation,cygan2008exponentialtime} 
for some such work in graph coloring, as well as the closely related Independent 
Set problem.  
As a simple example, one can obtain an $r$-coloring of a 3-colorable graph $G$ 
in time $O^*(\exp(O(n/r)))$ as follows: partition $G$ into $r/3$ blocks of size 
$3n/r$ each and use an exact 3-coloring exponential-time algorithm on each block.
In general, this idea is surprisingly close to the best we can do: 
any $r$-approximation to Graph Coloring requires at least 
$\exp(n^{1-o(1)}/r^{1+o(1)})$ time, assuming ETH \cite{bundit_thesis}. 

Recently, Bansal et. al. \cite{bansal2019new} obtained an $r$-approximation for 
the Graph Coloring problem that runs in 
$O^*\left(\exp\left(\tilde{O}\left(n/(r\log r) + r \log^2 r\right)\right)\right)$ 
time. As a subroutine in this algorithm, they use and develop an 
$r$-approximation for the Independent Set problem that runs in \\
$O^*\left(\exp\left(\tilde{O}\left(n/(r\log^2 r) + r \log^2 r\right)\right)\right)$ time. 
In this work we improve upon the approximation algorithms given by Bansal et. al.
\cite{bansal2019new} for the setting in which the input graph $G$ is known to be 
3-colorable.

\subsection{Promise Constraint Satisfaction Problems and Width-based Algorithms}

Graph colorability is an example of a larger class of problems known as 
\emph{constraint satisfaction problems} (CSPs). The input to a CSP is a set of 
variables $X$ over a domain $D$, along with a set of constraints. 
Each constraint specifies, for some subset of variables, a relation that it must 
takes values in. The goal is to find a valid assignment of values from the 
domain to each variable that satisfies all constraints. 
Equivalently, one can model CSPs as a problem of finding homomorphism 
between relational structures \cite{feder_vardi98}. That is, we can think of the 
input as two structures: $\mathbb{X} = (X, S_1, \dots, S_n)$, 
where $S_i \subseteq X^{r_i}$, $\mathbb{D} = (D, R_1, \dots, R_n)$, 
where $R_i \subseteq D^{r_i}$. The goal is to find a function $f: X \to D$ 
such that $f((x_1), \dots, f(x_{r_i})) \in R_i$ for all 
$(x_1, \dots, x_{r_i}) \in S_i$. In the case of $r$-graph colorability, 
the variables are the vertices of the graph, the sets $S_i$ are the edges, 
the domain is the set of colors $\{1, \dots, r\}$, and the constraints $R_i$ are 
all tuples of two different colors. 

Brakensiek and Guruswami \cite{brakensiek_guruswami21}, building on work in \cite{AGH17},
recently proposed a 
generalization of this class of problems, known as 
\emph{promise constraint satisfaction problems} (PCSPs); these are CSPs in which 
we are guaranteed that the input satisfies an even stronger set of constraints. 
Formally, we are given three relational structures 
$\mathbb X, \mathbb A, \mathbb B$ where there exists a homomorphism from 
$\mathbb B$ to $\mathbb A$. Given the promise that there exists some homomorphism 
from $\mathbb X$ to $\mathbb B$, the goal is to find one from $\mathbb X$ to 
$\mathbb A$. It is easy to see that the problem of $r$-coloring a 3-colorable 
graph is an example of a promise CSP, with $\mathbb A$ 
corresponding to $r$-colorability and $\mathbb B$ to 3-colorability.

Fixed-template CSPs are CSPs in which the structure $\mathbb A$ is fixed and the 
input specifies $\mathbb X$. One of the few known techniques for solving fixed-template CSPs in 
polynomial time is \emph{local consistency checking}. This technique considers 
\emph{$k$-strategies}, which are sets of partial solutions on subsets of variables 
of size at most $k$. These sets must be closed under restrictions 
and all partial solutions of subsets of size less than $k$ must be extendable to 
all supersets up to size $k$. 
Clearly, if there exists a homomorphism from $\mathbb X$ to $\mathbb A$, then 
there exists a $k$-strategy for all $k \leq |X|$. A CSP is said to have \emph{width} $k$ if the 
converse holds: a non-empty $k$-strategy implies a homomorphism from 
$\mathbb X$ to $\mathbb A$. For any fixed $k$, this gives a simple polynomial 
time algorithm to solve the CSP: find all partial solutions on subsets of 
size at most $k$ and remove the ones that don't satisfy the closure and 
extension properties until we either have a non-empty $k$-strategy or an 
empty set. 

One can define a similar notion in the context of PCSPs \cite{blko21}. 
A PCSP has width $k$ if a non-empty $k$-strategy on $(\mathbb X, \mathbb B)$ 
implies a homomorphism from $\mathbb X$ to $\mathbb A$. Recently, 
Atserias and Dalmau \cite{atserias2022promise} showed that given an input graph 
$G$ and parameter $\eps \in (0, 1/2)$, the problem of distinguishing whether 
$G$ is 3-colorable from the case where $G$ is not $O(n^{\eps})$-colorable 
has width between $n^{1-2\eps}$ and $n^{1-3\eps}$. This gives an algorithm to 
color 3-colorable graphs with $O(n^\eps)$ colors in time 
$2^{\Theta(n^{1-2\eps}\log n)}$. Improving this runtime is posed as an open 
question in \cite{atserias2022promise}.

\subsection{Definitions and Notation}

For a given graph $G$, we let $V(G)$ denote its set of vertices, 
$E(G)$ denote the set of edges, and $n := |V(G)|$. 
For any $v \in V(G)$, we call the \emph{neighborhood} of $v$ 
(denoted $N_G(v)$) the set of all $u \in V(G)$ such that $\{u, v\} \in E(G)$. 
We let $d_G(v) = |N_G(v)|$ denote the \emph{degree} of $v$. 
We write $O^*$ to suppress polynomial factors in $n$, and $\tilde{O}$ to 
suppress $\polyloglog(r)$ factors (where $r$ is the approximation ratio). 
The analogous definitions hold for $\tilde{\Omega}$. 

\mypara{Maximum Independent Set}: We say a set $S \subseteq V(G)$ is 
\emph{independent} if for all $u, v \in S$, $\{u, v\} \notin E$. 
We let $\alpha(G)$ denote the size of the maximum cardinality independent set 
of $G$. 

\mypara{Graph Coloring}: We say that a function $c: V(G) \to [r]$ is a 
\emph{$r$-coloring} if $c(u) \neq c(v)~\forall \{u, v\} \in E(G)$. 
We say $G$ is \emph{$r$-colorable} if a valid $r$-coloring of $G$ exists. 
We let $\chi(G)$ denote the chromatic number of $G$, i.e. the minimum 
integer $r$ such that $G$ is $r$-colorable.

\subsection{Our Result and Techniques}

Our main contribution is an improved appoximation algorithm for coloring 
3-colorable graphs.

\begin{theorem}
  \label{thm:main_result}
  Given a 3-colorable graph $G = (V, E)$ and a parameter $r > 0$, 
  there exists an algorithm that, with high probability, outputs a 
  valid coloring of $G$ using at most $O(r)$ colors in time 
  $O^*\left(\exp\left(\tilde{O}\left(\frac {n\log^{11/2}r} {r^3}\right)\right)\right)$. 
\end{theorem}

For any fixed real $\eps \in (0, 1/3)$, setting $r$ to be $n^{\eps}$, 
we obtain a coloring of $G$ using $O(n^{\eps})$ colors in time 
$O^*\left(\exp\left(O(n^{1-3\eps})\right)\right)$. 
This asymptotically improves 
upon the width-based algorithm given by Atserias and Dalmau \cite{atserias2022promise}. 
Furthermore, up to polynomial factors in $n$, this matches the width-based 
lower bound given in~\cite{atserias2022promise}. Whether our result can be 
captured by a width-based argument remains a direction for future research.

\mypara{Techniques:} We closely follow the algorithm given by Bansal et al. 
\cite{bansal2019new}. We first consider the independent set problem.
The approach in \cite{bansal2019new} is 
to branch on high degree vertices. For each such vertex, the algorithm chooses 
to include it with some small probability, in which case all of its neighbors are 
removed, else it removes the high degree vertex and recurses on the remaining 
graph. The probability of choosing each vertex is chosen in a way to ensure the 
number of leaves of the branching process remains bounded. Once the maximum 
degree of the graph is small, the algorithm terminates 
by applying a known approximation for for graphs with bounded degree \cite{bgg15}.
We follow from a similar algorithm, replacing the base case 
with an $\tilde \Omega(d^{1/3})$-approximation for Independent Set on 
3-colorable graphs. 
To then color the graph $G$, the algorithm in~\cite{bansal2019new} repeatedly 
finds large independent sets and colors each independent set with its own color. 
Once the number of remaining vertices is small enough, they use a brute force 
inclusion-exclusion based algorithm for coloring the remaining vertices in time 
$O^*(2^{|V|})$~\cite{BHK09}. We follow the same approach, using our improved 
approximation for Independent Set.

\section{Improved Approximation for Coloring 3-Colorable Graphs}

\subsection{Independent Set}
\label{sec:max_is}

In this section, we prove the following lemma providing an approximation algorithm 
for finding an independent set in a 3-colorable graph:

\begin{lemma}
  \label{lem:is_result}
  Given a 3-colorable graph $G = (V, E)$ and a parameter $r \in \Z_{>0}$, 
  there exists an algorithm that, with probability at least $1/2$, 
  outputs an independent set of size at least $\alpha(G)/r$ in time 
  $O^*\left(\exp\left(\tilde O\left(\frac{n \log^{5/2}r}{r^3}\right)\right)\right)$. 
\end{lemma}

We use the following lemma from \cite{bansal2019new}.

\begin{lemma}[\cite{bansal2019new}]
\label{lem:bansal_indset}
  Let $\calA(G')$ be an approximation algorithm that outputs an independent set of 
  $G'$ of size $\alpha(G')/r$. Suppose $\calA(G')$ runs in time 
  $T_\calA(n', r)$ (where $n' = |V(G')|$) 
  when $G$ has maximum degree $d(r)$, where $d(r) \geq 2r$. 
  Then, there is an algorithm that given any graph $G$, outputs an independent 
  set of $G$ of expected size $\alpha(G)/r$ in expected time 
  $O^*\left(\exp\left(\frac n {d(r)} \log(4d(r)/r)\right)T_\calA(n,r)\right)$.
\end{lemma}

We note that the algorithm of \cite{bansal2019new} only calls 
$\calA(G')$ on graphs $G'$ obtained by deleting vertices and edges from $G$. 
Therefore, if $G$ is 3-colorable, we can assume $G'$ is as well.
The seminal work of Karger, Motwani, and Sudan~\cite{karger1994approximate} 
provides a polynomial time algorithm to color 3-colorable graphs using 
$O(d^{1/3}\log^{1/2}d \log n)$ colors, where $d$ is the maximum degree of the 
input graph. We will use the following lemma from their analysis:
\begin{lemma}[\cite{karger1994approximate}]
\label{lem:karger_sdp}
  Given a $3$-colorable graph $G$, there exists a polynomial time algorithm to 
  find an independent set of size $\Omega(n/(d^{1/3}\sqrt{\log d}))$,
  where $d$ is the maximum degree of the input graph.
\end{lemma}

Combining these two results, we obtain the following corollary:

\begin{corollary}
\label{cor:is} 
  Given a $3$-colorable graph $G$, there exists a polynomial time algorithm 
  that outputs an independent set of expected size $\alpha(G)/r$ in expected 
  time $O^*\left(\exp\left(\tilde{O}\left(\frac {n\log^{5/2}r} {r^3}\right)\right)\right)$. 
\end{corollary}
\begin{proof}
  We refer to the algorithm given by Lemma \ref{lem:karger_sdp} as $\calA$.
  By Lemma \ref{lem:karger_sdp}, 
  for any $3$-colorable graph $G'$ with maximum degree $d = \frac{r^3}{\log^{3/2}r}$, 
  $\calA(G')$ outputs an independent set of size 
  \[\Omega\left(\frac{n}{d^{1/3}\sqrt{\log d}}\right) = 
  \Omega\left(\frac{n \sqrt{\log(r)}}{r\sqrt{\log (r /\log r)}}\right).\] 
  Note that since $G'$ is $3$-colorable, $\alpha(G') \geq n/3$. Thus 
  we can choose $d(r) = \tilde O(r^3/\log^{3/2}r)$ so that 
  $\calA(G')$ outputs an independent set of size $\alpha(G')/r$ in 
  polynomial time (recall here that $\tilde O$ suppresses 
  $\polyloglog(r)$ factors). 

  We then apply Lemma \ref{lem:bansal_indset} on the input $G$ 
  with $\calA(G')$ as defined above to obtain an independent set of expected 
  size $\alpha(G)/r$ in expected time 
  \begin{align*}
    O^*\left(\exp\left(\frac {n \log^{3/2}r} {r^3} \log\left(\frac{4r^3}{r\log^{3/2}r}\right)\right)\poly(n)\right) 
        = O^*\left(\exp\left(\tilde{O}\left(\frac {n\log^{5/2}r} {r^3}\right)\right)\right). 
  \end{align*}
\end{proof}

To complete the proof of Lemma \ref{lem:is_result}, note that Corollary 
\ref{cor:is} gives the desired runtime and approximation ratio in expectation. 
Thus it suffices to run multiple iterations of this algorithm to obtain 
our desired probability of failure.

\begin{proof}[Proof of Lemma \ref{lem:is_result}]
  Let $G$ be a given 3-colorable graph. Let $\calB(G, r)$ denote the algorithm 
  given by Corollary \ref{cor:is}, and let $T'(n,r) = 
  O^*\left(\exp\left(\tilde O \left(\frac {n\log^{5/2}r} {r^3}\right)\right)\right)$ 
  denote the expected runtime of $B(G, r)$. 
  We run $B(G, r/2)$ $r$ times and output the maximum independent set given. 
  If the total runtime exceeds $5rT'(n, r/2)$, we terminate the algorithm.

  By construction, this algorithm runs in time $O(r T'(n, r/2)) = 
  O^*\left(\exp\left(\tilde O\left(\frac {n\log r} {r^3}\right)\right)\right)$. 
  By Markov's inequality, the probability that the total runtime exceeds 
  $5rT'(n, r/2)$ is at most $\frac 1 5$. Thus it remains to bound the error 
  probability given that we successfully complete $r$ iterations of $B(G, r/2)$.

  Fix a single run of $B(G, r/2)$, and let $S$ be the resulting independent set. 
  To bound the probability that $S$ is small, note that 
  $\Pr[|S| \leq \alpha(G)/r] = \Pr[\alpha(G) - |S| \geq \alpha(G)(1 - 1/r)]$. 
  Since $|S| \leq \alpha(G)$, we can apply Markov's inequality to obtain 
    \begin{align*}
      \Pr[\alpha(G) - |S| \geq \alpha(G)(1 - 1/r)] \leq 
      \frac{\E[\alpha(G) - |S|]}{\alpha(G)(1 - 1/r)} = 
      \frac{\alpha(G)(1 - 2/r)}{\alpha(G)(1 - 1/r)} = 1 - \frac{1}{r-1}.
    \end{align*}
  Thus the probability that none of the $r$ runs of $B(G, r/2)$ give an 
  independent set of size at least $\alpha(G)/r$ is at most 
  $(1-1/(r-1))^r \leq 1/e$. 
  Combining with the runtime error, we get a total probability of failure of at 
  most $\frac 1 5 + \frac 4 {5e} < 1/2$.
\end{proof}

\subsection{Graph Coloring}

In this section, we refer to the algorithm given by Lemma \ref{lem:is_result} 
as $\IS(G, r)$ for consistency with the notation used in~\cite{bansal2019new}.
To prove Theorem \ref{thm:main_result}, we follow the same general approach 
outlined by Bansal et. al. in \cite{bansal2019new} with $\IS(G, r)$ as a 
subroutine. For completeness, we describe the full algorithm and analysis below. 
Let $t = \frac n {r^3}$, $r' = r/\log(n/t) = \frac r {\log r^3}$.

\begin{algorithm}[H]
\caption{$r$-Approximate Graph Coloring}
\label{alg:graph_col}
\begin{algorithmic}[1]
  \State $c \gets 0$
  \While{$|V| \geq t$} \label{algo:line:loop-start}
    \State $c \gets c + 1$
    \State $C_c \gets \emptyset$
    \For{$i = 1, \dots, n$}
      \State $C_c \gets \textnormal{argmax}(|C_c|, |\IS(G[V], r')|)$ 
      \label{algo:line:is}
    \EndFor
    \State $V \gets V \setminus C_c$
  \EndWhile \label{algo:line:loop-end}
  \State $(C_{c+1}, C_{c+2}, C_{c+3}) \gets$ brute force optimum coloring of $G[V]$ \label{algo:line:brute-force}\\
  \Return $(C_1, \dots, C_{c+3})$
\end{algorithmic}
\end{algorithm}

\begin{claim}
\label{claim:is_boosting}
  For a fixed iteration $c$ of Algorithm \ref{alg:graph_col}, 
  $C_c \geq \frac{|V|}{3r'}$ with probability at least $1 - (1/2)^n$. 
\end{claim}
\begin{proof}
  At iteration $c$, $G[V]$ must have an independent set of size at least $|V|/3$, 
  since $G$ is 3-colorable and removing vertices cannot increase the chromatic 
  number. 
  By Lemma \ref{lem:is_result}, for each $i \in [n]$, 
  the independent set $C_c$ given by line \ref{algo:line:is}, 
  $|C_c| \geq \alpha(G[V])/r' \geq \frac {|V|} {3r'}$ with probability at 
  least $\frac 1 2$. Thus $\Pr[|C_c| < |V|/(3r')] \leq (1/2)^n$ as desired.
\end{proof}

\begin{claim}
  \label{claim:correctness}
  Algorithm \ref{alg:graph_col} gives a valid coloring of $G$ using at most 
  $O(r)$ colors with high probability. 
\end{claim}
\begin{proof}
  Validity of the coloring follows is clear by construction, since we choose a 
  new color for each independent set. It remains to bound the number of 
  colors used.

  Note that $c \leq n$ at all times, since $|V|$ reduces by at least one at each 
  iteration. Thus 
  by Claim \ref{claim:is_boosting}, $|V|$ reduces by a factor of $(1 - 1/(3r'))$ 
  at every iteration with probability at least $1 - n/(2^n)$. 
  After $\ell$ iterations, 
    \[|V| \leq n\left(1 - \frac 1 {3r'}\right)^{\ell}
    \leq n \cdot \exp\left(-\frac \ell {3r'}\right) 
    = n \cdot \exp\left(-\frac{\ell \ln(n/t)}{3r}\right) 
    = n \cdot \left(\frac n t\right)^{-\ell/3r}.\]
  Therefore, when $\ell = 3r$, $|V| \leq t$, so Algorithm \ref{alg:graph_col} 
  runs for at most $3r+1$ iterations. Since $G$ is 3-colorable, 
  the final optimum coloring of the remaining vertices 
  on line \ref{algo:line:brute-force} uses at most 3 colors. 
  Thus the total number of colors used is at most $3r+4 = O(r)$.
\end{proof}

\begin{claim}
  \label{claim:runtime}
  Algorithm \ref{alg:graph_col} runs in time 
  $O^*\left(\exp\left(\tilde{O}\left(\frac {n\log^{11/2}r} {r^3}\right)\right)\right)$.
\end{claim}
\begin{proof}
  We first bound the runtime of the while loop in lines 
  \ref{algo:line:loop-start}-\ref{algo:line:loop-end}. 
  By Lemma \ref{lem:is_result}, $\IS(G[V], r')$ runs in time 
  $O^*\left(\exp\left(\tilde{O}\left(\frac {n\log^{5/2}r'} {r'^3}\right)\right)\right)$. 
  By the analysis in the proof of Claim \ref{claim:correctness}, 
  the total number of iterations is at most $O(r)$, and in each iteration 
  we run $\IS(G[V], r')$ $r$ times. This the loop runs in time 
  \begin{align*}
    &O^*\left(r^2\exp\left(\tilde{O}\left(\frac {n\log^{5/2}r'} 
    {r'^3}\right)\right)\right) 
    = O^*\left(\exp\left(\tilde{O}\left(\frac {n\log^{5/2}(r/\log(r^3))} 
    {r^3/\log^3(r^3)}\right)\right)\right) \\
    = &O^*\left(\exp\left(\tilde O\left(\frac{n \log^{11/2} r}{r^3}\right)\right)\right).
  \end{align*}

  Using a naive brute force algorithm, the final step in line \ref{algo:line:brute-force}
  runs in time 
  $O\left(3^t\right) = O\left(\exp\left(O\left(\frac{n}{r^3}\right)\right)\right)$, 
  giving us our desired time bound.
\end{proof}

\bibliographystyle{plainurl}
\bibliography{references}

\end{document}